\documentclass[journal,12pt,onecolumn]{elsarticle}

\usepackage{amsmath,amsfonts}
\usepackage{amssymb,amsthm}
\usepackage{makecell}
\usepackage{hyperref}
\usepackage{graphicx}
\usepackage{titlesec}
\titlespacing{\subsection}{0pt}{1ex}{0pt}
\newtheorem{theorem}{Theorem}[section]
\newtheorem{lemma}{Lemma}[section]
\newtheorem{definition}{Definition}[section]
\newtheorem{corollary}{Corollary}[section]
\newtheorem{proposition}{Proposition}[section]
\newtheorem{example}{Example}[section]
\newtheorem{remark}{Remark}[section]
\newtheorem{conjecture}{Conjecture}[section]
\newtheorem{problem}{Problem}[]

\begin{document}
%\title{Dual and Covering Radii of Extended Algebraic Geometry Codes}
%\author{Yunlong Zhu\thanks{Y. Zhu is with the Department of Mathematics, School of Mathematics, Sun Yat-sen University, Guangzhou 510275, China (e-mail: zhuylong3@mail2.sysu.edu.cn).},
%Chang-An Zhao\thanks{C.-A. Zhao is with the School of Mathematics, Sun Yat-sen University, Guangzhou 510275, China, and also with the Guangdong Key Laboratory of Information Security, %Guangzhou 510006, China (e-mail: zhaochan3@mail.sysu.edu.cn).}}
%\date{\today}
%\maketitle
%\begin{abstract}
%{\bf Index Terms:} 
%\end{abstract}

\begin{frontmatter}
\title{Dual and Covering Radii of Extended Algebraic Geometry Codes}
\author[sysu]{Yunlong Zhu}
\ead{zhuylong3@mail2.sysu.edu.cn}
\author[sysu,ist]{Chang-An Zhao\corref{cor1}}
\cortext[cor1]{Corresponding author}
\ead{zhaochan3@mail.sysu.edu.cn}
	
\fntext[fn1]{This work is supported  by  Guangdong Basic and Applied Basic Research Foundation of China~(No. 2025A1515011764), the National Natural Science Foundation of China (No. 12441107), Guangdong Major Project of Basic and Applied Basic Research (No. 2019B030302008) and Guangdong Provincial Key Laboratory of Information Security Technology(No. 2023B1212060026). }
	
\address[sysu]{Department of Mathematics, School of Mathematics, Sun Yat-sen University, Guangzhou 510275, P.R.China.}
\address[ist]{Guangdong Key Laboratory of Information Security Technology, Guangzhou 510006,  P.R.China.}

\date{\today}	
	
\begin{abstract}
Many literatures consider the extended Reed-Solomon (RS) codes, including their dual codes and covering radii, but few focus on extended algebraic geometry (AG) codes of genus $g\ge1$. In this paper, we investigate extended AG codes and Roth-Lempel type AG codes, including their dual codes and minimum distances. Moreover, we show that for certain $g$, the length of a $g$-MDS code over a finite field $\mathbb{F}_q$ can attain $q+1+2g\sqrt{q}$, which is achieved by an extended AG code from the maximal curves of genus $g$. Notably, for some small finite fields, this length $q+1+2g\sqrt{q}$ is the largest among all known $g$-MDS codes. Subsequently, we establish that the covering radius of an $[n,k]$ extended AG code has $g+2$ possible values. For the case of $g=1$, we prove that this range reduces to two possible values when the length $n$ is sufficiently large, or when there exists an $[n,k+1]$ MDS elliptic code.
\end{abstract}
	
\end{frontmatter}

\section{Introduction}
Let $\mathbb{F}_q$ be a finite field with $q$ elements. An $[n,k]$-linear code $\mathcal{C}$ is a $k$-dimensional subspace of $\mathbb{F}_q^n$. For a non-zero vector ${\bf v}=(v_1,\ldots,v_n)\in\mathbb{F}_q^n$, the Hamming weight $wt({\bf v})$ is defined as the number of non-zero positions in ${\bf v}$, i.e.,
\[
	{\rm wt}({\bf v}):=\#\{i\mid v_i\neq0,1\le i\le n\}.
\]
The distance between two vectors ${\bf v}$ and ${\bf u}$ is given by $d({\bf v},{\bf u})={\rm wt}({\bf v-u})$ and the minimum distance of $\mathcal{C}$ is 
\[
	d:=\min\{d({\bf v},{\bf u})|{\bf v,u}\in \mathcal{C}\}.
\]

For an $[n,k,d]$-linear code $\mathcal{C}$, the well-known Singleton bound states that
\[
	d+k\le n+1.
\]
Linear codes achieving this equality are known as maximum distance separable (MDS) codes. The Singleton defect \cite{Boer96} of $\mathcal{C}$, denoted $s(\mathcal{C})$, is defined as 
\[
	s(\mathcal{C})=n+1-k-d.
\]
A code $\mathcal{C}$ is called $\ell$-MDS if $s(\mathcal{C})=s(\mathcal{C}^{\perp})=\ell$ \cite{Li-Zhu-NMDS}. In particular, a $1$-MDS code, is also called a near-MDS (NMDS) code. A central problem in coding theory involves determining the maximum possible length of $\ell$-MDS codes. We recall the famous MDS conjecture proposed by Segre \cite{Segre}:
\begin{conjecture}\label{serges}
Let $\mathcal{C}$ be an $[n,k]$ MDS code over $\mathbb{F}_q$. Then
\begin{align*}
n\le\begin{cases}
q+2\ \text{\rm if $q$ is even and $k\in\{3,q-1\}$},\\
q+1\ \text{\rm otherwise}.
\end{cases}
\end{align*}
\end{conjecture}
\noindent While this conjecture remains open in general, partial results can be found, for example, in \cite{Ball-prime,Ball-ld,Walker-conjecture}.

The covering radius $\rho(\mathcal{C})$ of a linear code $\mathcal{C}$ is the smallest integer such that spheres of radius $\rho(\mathcal{C})$ centered at all codewords cover $\mathbb{F}_q^n$. This concept represents a fundamental parameter in coding theory \cite{Cohen1997}, with many further investigations in \cite{Zhang-Wan-deephole,Cohen1985,Helleseth1978,Hou92,Hou93,Graham1985}. Vectors at distance $\rho(\mathcal{C})$ from $\mathcal{C}$ are called deep holes. Notably, deep holes of MDS codes preserve the MDS property when appended to their generator matrices, as proven by Kaipa \cite{Kaipa-deepERS} and Wu et al. \cite{Wuyansheng-when}.

Both the minimum distance and the covering radius determination for general linear codes are NP-hard problems \cite{Vardy97,McLoughlin}. For algebraic geometric (AG) codes, the minimum distance satisfies $d\ge d^*$ where $d^*$ denotes the designed distance. However, bounds on the covering radii remain scarce; we refer the reader to \cite{Orozco-Janwa-CR}, \cite{Janwa90} for some results. When $n\le q$, Zhuang et al. proposed that the Reed-Solomon (RS) codes over $\mathbb{F}_q$ achieve a covering radius $\rho=n-k$ \cite{Zhuang-deepGRS}. Bartoli et al. demonstrated that there exist specific MDS elliptic codes with $\rho=n-k-1$ \cite{Bartoli-covering}. Zhang and Wan further showed that elliptic codes also have $\rho=n-k-1$ for sufficiently large $n$ \cite{Zhang-Wan-deephole}. Nevertheless, the covering radii of general MDS codes remain largely undetermined \cite{Zhang-Wan-Kaipa-PRS}.

\subsection*{Motivations and Contributions}
Extended AG codes from curves of genus $g$ are observed to achieve the maximum length among all known $g$-MDS codes over small finite fields, as recorded in the MAGMA database \cite{Magma}. This motivates investigating their maximality over arbitrary finite fields $\mathbb{F}_q$. While the existing literature primarily focuses on the genus $0$ case, extended AG codes with $g \geq 1$ remain largely unexplored, to the best of our knowledge. 

For $g=0$, extended RS codes are the well-known MDS codes with length $n=q+1$. Under the MDS conjecture, these codes become maximal MDS codes when $4 \leq k \leq q-2$. Zhang et al. investigated the covering radii of extended RS codes \cite{Zhang-Wan-Kaipa-PRS}, and showed that these codes have two possible values. Very recently, Wu et al. proposed improved results for extended RS codes \cite{Wuyansheng-when}, while Wu et al. further determined the covering radii of non-standard extended RS codes \cite{Wu2024mds}, including Roth-Lempel codes \cite{Roth-nRS}. Additionally, Li et al. analyzed the covering radii of two extended twisted generalized RS codes \cite{Li-Zhu-ETGRS}.

Consequently, a natural problem arises for the case of $g\ge1$: can one determine the dual codes, minimum distances and covering radii of extended AG codes? Our main results are summarized as follows.
\begin{itemize}
	\item Utilizing the Weil differential, we explicitly determine the generator matrix of the dual codes for extended AG codes. By establishing an isomorphism between a function space and a differential space over $\mathbb{F}_q$ via a canonical divisor of function fields, we propose a function-based generator matrix for the dual code of an extended AG code, which generalizes the result for standard AG codes. Subsequently, we explicitly derive the lower bound of the minimum distances for extended AG codes, and investigate the MDS property of extended AG codes. Furthermore, we consider three cases of AG codes, from the projective line, elliptic curves and Hermitian curves respectively. We then present the explicit dual code and minimum distances for these codes.
	\item We explicitly determine dual code generators for Roth-Lempel type AG codes, providing both differential-based and function-based generator matrices. A sufficient condition for the minimum distance of these codes is established. For elliptic curves, we further establish a condition for Roth-Lempel-type elliptic codes to be NMDS.
	\item We adapt covering radius bounds from \cite{Janwa90} to extended AG codes. Furthermore, for elliptic curves, we prove that the covering radius of extended elliptic codes has two possible values when $n$ is sufficiently large or when an $[n,k+1]$ MDS elliptic code exists.
\end{itemize}
\subsection*{Organization}
This paper is organized as follows. In Section II, we provide an essential overview of algebraic geometry codes, elliptic curves and covering radii. Section III details our main results on extended AG codes and Roth-Lempel type AG codes. In Section IV, we discuss the covering radius bounds for extended AG codes. Finally, Section V concludes the paper and discusses our future research directions.

\section{Preliminaries}
In this section, we review fundamental definitions and concepts, including AG codes and covering radii for linear codes.
\subsection{Algebraic Geometry Codes}
Let $\mathcal{X}$ be an absolutely irreducible smooth algebraic curve defined over $\mathbb{F}_q$ with genus $g(\mathcal{X})$. The function field $F(\mathcal{X})$ of $\mathcal{X}$ is a finite extension of the rational function field $\mathbb{F}(x)$ where $x$ is transcendental over $\mathbb{F}_q$.

The set $\mathcal{X}(\mathbb{F}_q)$ of rational points corresponds to degree-one places of $F(\mathcal{X})$. A divisor $G$ is expressed as a formal sum of places:
\[
	G=\sum\limits_Pn_{P}P,
\]
where
\[
 \deg(G)=\sum\limits_Pn_{P}\ \text{and}\ \text{\rm supp}(G)=\{P | v_{P}(G)\neq0\}.
\]
For a function $f \in \mathbb{F}(\mathcal{X})$, the principal divisor is denoted by $(f)$. The Riemann-Roch space associated with a divisor $G$ is given by 
\[
	\mathcal{L}(G) := \{f\in F(\mathcal{X})\setminus\{0\}: (f)+G\ge 0\} \cup \{0\},
\]
with dimension $\ell(G)$ over $\mathbb{F}_q$. Let $P_1,\ldots,P_n$ be distinct degree-one places of $F(\mathcal{X})$ not in $\text{\rm supp}(G)$ with $n\le \#\mathcal{X}(\mathbb{F}_q)$. Let $D=P_1+\cdots+P_n$ and consider the evaluation map:
\begin{align*}
   \text{\rm ev}_D:  \mathcal{L}(G)&\to \mathbb{F}_q^n,\\
       f&\mapsto (f(P_1),\ldots,f(P_n)).
\end{align*}
The AG code $C_L(D,G)$, is defined as the image of $\text{\rm ev}_D$. The parameters of $C_L(D,G)$ are given by:
\[
	k=\ell(G)-\ell(G-D), d\ge n-\deg(G).
\]
It can be verified straightforwardly that $\text{\rm ev}_D$ is an embedding when $\deg(G)<n$ and $k=\ell(G)$. By the Riemann-Roch theorem, we obtain
\[
    \ell(G)= \deg(G)-g(\mathcal{X})+1,
\]
provided $\deg(G)\ge 2g(\mathcal{X})-1$. Moreover, the minimum distance satisfies $d\ge d^*$ where
\[
    d^*=n-k-g(\mathcal{X})+1
\]
is called the designed distance of $C_L(D,G)$.
\begin{lemma}\label{generator}\cite[Theorem 13.4.3]{Huffman-fecc}
Let $\{f_1,f_2,\ldots,f_{\ell}\}$ be a basis of $\mathcal{L}(G)$. Then the matrix
\begin{equation*}
	G:=\begin{pmatrix}
		f_1(P_1)& f_1(P_2)&\cdots &f_1(P_n)\\
		f_2(P_1)& f_2(P_2)&\cdots &f_2(P_n)\\
		\vdots&\vdots&\ddots&\vdots\\
		f_{\ell}(P_1)& f_{\ell}(P_2)&\cdots &f_{\ell}(P_n)\\
	\end{pmatrix}
\end{equation*}
is a generator matrix for $C_L(D,G)$.
\end{lemma}

Another code associated with divisors $D$ and $G$ is $C_{\Omega}(D,G)$, defined by:
\[
    C_{\Omega}(D,G):=\left\{(\omega_{P_1}(1),\ldots,\omega_{P_n}(1)) | \omega\in\Omega(G-D) \right\}
\]
where $\Omega(G-D)$ denotes the Weil differential space with the definition,
\[
	\Omega(G):=\{\omega\in \Omega_F(\mathcal{X}): (\omega)+G\ge 0\} \cup \{0\}.
\]
The divisor $(\omega)$ is called a canonical divisor. The code $C_{\Omega}(D,G)$ is also an AG code. The relationship between $C_L(D,G)$ and $C_{\Omega}(D,G)$ is described in the following (\cite[Proposition 2.2.10]{Stich}):
\begin{proposition}\label{dual}
Let $\eta$ be a Weil differential such that $v_{P_i}(\eta)=-1$ and $\eta_{P_i}(1)=1$ for $i=1,\ldots,n$. Then
\[
       C_L(D,G)^{\perp}=C_{\Omega}(D,G)=C_L(D,H)
\]
where
\[
	H:=D-G+(\eta).
\]
\end{proposition}
Furthermore, we present a useful isomorphism between function spaces and differential spaces.
\begin{lemma}\cite[Theorem 1.5.14]{Stich}\label{isomorphism}
Let $G$ be a divisor, and $W$ be a canonical divisor with $W=(\omega)$. Then the mapping
\begin{align*}
\mu:\begin{cases}
\mathcal{L}(W-G) \to \Omega(G),\\
x\to x\omega
\end{cases}
\end{align*}
is an isomorphism of $\mathbb{F}_q$-vector spaces.
\end{lemma}

Let $\mathcal{X}$ be an elliptic curve and $P_1,\ldots,P_n,P_{\infty}$ be degree-one places of the function field $F(\mathcal{E})$. By the Singleton bound and the designed distance, the minimum distance of $C_L(D,mP_{\infty})$ lies in $[n-m,n-m+1]$. It is well known that if $d=n-m$, then $C_L(D,mP_{\infty})$ is NMDS; otherwise, it is MDS. Determining the minimum distance of $C_L(D,mP_{\infty})$ is equivalent to solving a subset sum problem, which is generally intractable. Very recently, Han and Ren showed the following result \cite{Han-conjecture}:
\begin{lemma}\label{maxecmds}
Let $C_L(D,mP_{\infty})$ be an MDS elliptic code, then
\[
	n\le\frac{\#\mathcal{E}(\mathbb{F}_q)}{2}+3.
\]
\end{lemma}

\subsection{Covering Radii of Linear Codes}
Let $\mathcal{C}$ be an $[n,k,d]$ linear code over $\mathbb{F}_q$. For any ${\bf v}\in\mathbb{F}_q^n$, define
\[
	d({\bf v},\mathcal{C}):=\min\{d({\bf v},{\bf c})|{\bf c}\in \mathcal{C}\}.
\]
The covering radius is then defined as follows.

\begin{definition}\cite{Huffman-fecc}\label{defcd}
The covering radius $\rho(\mathcal{C})$ of $\mathcal{C}$ is the smallest integer $\rho$ such that every vector in $\mathbb{F}_q^n$ has distance at most $\rho$ from some codeword in $\mathcal{C}$. Equivalently,
\[
	\rho(\mathcal{C}):=\max\limits_{{\bf v}\in\mathbb{F}_q^n}d({\bf v},\mathcal{C}).
\]
The vectors achieving $\rho(\mathcal{C})$ are called deep holes of $\mathcal{C}$.
\end{definition}
We recall two lemmas on covering radii that will be used subsequently.
\begin{lemma}\cite[Redundancy Bound]{Huffman-fecc}
For an $[n,k]$ code $\mathcal{C}$, $\rho(C)\le n-k$.
\end{lemma}
Suppose that $n_q(k,d):=\min\{n; \text{there is an $[n, k, d]\ q$-ary code}\}$. We have:
\begin{lemma}\cite[Corollary 8.1 ]{Orozco-Janwa-CR}\label{oplength}
For an $[n_q(k,d),k,d]$ code $\mathcal{C}$, $\rho(\mathcal{C})\le d-\frac{d}{q^k}$.
\end{lemma}
\section{The Dual of Extended Algebraic Geometry Codes}
In this section, we investigate extended AG codes and Roth-Lempel type AG codes. We present the generator matrices for their dual codes and analyze minimum distances. Firstly, we fix the following notations for the remainder of this paper:
\begin{itemize}
	\item $\mathcal{X}$ denotes an algebraic curve over $\mathbb{F}_q$ with genus $g$.
	\item $F=F(\mathcal{X})$ represents the algebraic function field of $\mathcal{X}$.
	\item $n$ and $m$ are integers satisfying	 $2g\le m\le n-1$.
	\item $D=P_1+\cdots+P_n$ and $mP_{\infty}$ are effective divisors, with $P_i,P_{\infty}$ being degree-one places of $F$.
	\item $k=m-g+1$, where $\{f_1,\ldots,f_{k}\}$ forms a basis for $\mathcal{L}(mP_{\infty})$.
	\item $C_L(D,mP_{\infty})$ denotes the AG code associated with $D$ and $mP_{\infty}$, having generator matrix $G_k$ and parity check matrix $H_k$.
\end{itemize}

\subsection{Codes with Length \texorpdfstring{$n+1$}{}}
The extended code $C_{ex}(D,mP_{\infty})$ is generated by the matrix
\begin{equation*}
	G_{k,ex}:=\begin{pmatrix}
		G_k & \infty^T\\
	\end{pmatrix},
\end{equation*}
where
\[
	\infty=(0,0,\ldots,0,1).
\]
The main results will be given as follows.
\begin{theorem}
The code $C_{ex}(D,mP_{\infty})$ forms an $[n+1,k]$-linear code over $\mathbb{F}_q$. For any
\[
	\omega^{*}\in\Omega_F((m-1)P_{\infty}-D)\setminus\Omega_F(mP_{\infty}-D),
\]
the matrix
\begin{equation*}
	H_{k,ex}:=\begin{pmatrix}
		H_k & {\bf 0}\\
		{\bf w} & \omega^*_{P_{\infty}}(f_{k})\\
	\end{pmatrix}
\end{equation*}
serves as a parity check matrix for $C_{ex}(D,mP_{\infty})$, where
\[
	{\bf w}=(\omega^{*}_{P_1}(1),\omega^{*}_{P_2}(1),\ldots,\omega^{*}_{P_n}(1)).
\] 
\end{theorem}
\begin{proof}
The length and dimension of $C_{ex}(D,mP_{\infty})$ follow directly from the condition $2g\le m\le n-1$. We observe that
\[
	\dim_{\mathbb{F}_q}(\Omega_F((m-1)P_{\infty}-D))=n-m+g
\]
and
\[
	\dim_{\mathbb{F}_q}(\Omega_F(mP_{\infty}-D))=n-m+g-1.
\]
Hence
\[
	\Omega_F((m-1)P_{\infty}-D)\setminus\Omega_F(mP_{\infty}-D)\neq\emptyset.
\]
Given that $G_k$ and $H_k$ are the generator and parity check matrices of $C_L(D,mP_{\infty})$, respectively, it follows that
\begin{equation*}
	\begin{pmatrix}
		H_k & \mathbf{0}
	\end{pmatrix}\cdot
	\begin{pmatrix}
			G_k & \infty^T
	\end{pmatrix}^T=
	\begin{pmatrix}
			0
	\end{pmatrix}.
\end{equation*}
To complete the proof, it remains to show that
\begin{equation*}
	\begin{pmatrix}
		{\bf w} & \omega^*_{P_{\infty}}(f_{k})
	\end{pmatrix}\cdot
	\begin{pmatrix}
		G_k & \infty^T
	\end{pmatrix}^T=
	\begin{pmatrix}
		0
	\end{pmatrix}.
\end{equation*}
For $\omega^{*}\in\Omega_F((m-1)P_{\infty}-D)\setminus\Omega_F(mP_{\infty}-D)$, we have $v_{P_{\infty}}(\omega^*)=m-1$ and $v_{P_i}(\omega^*)=-1$. Applying the fact that Weil differentials vanish on $F$, for $1\le j\le k-1$, we obtain
\begin{align*}
	0&=\omega^*(f_j)\\
	&=\sum\limits_{i=1}^n\omega^*_{P_i}(f_j)\\
	&=\sum\limits_{i=1}^nf_j(P_i)\omega^*_{P_i}(1)\\
	&=\begin{pmatrix}
			{\bf w} & \omega^*_{P_{\infty}}(f_{k})
		\end{pmatrix}\cdot\begin{pmatrix}
					f_j(P_1)&\cdots&f_j(P_n) & 0
				\end{pmatrix}^T.
\end{align*}
For $f_{k}$, we derive
\begin{align*}
	0&=\omega^*(f_{k})\\
	&=\sum\limits_{i=1}^n\omega^*_{P_i}(f_{k})+\omega^*_{P_{\infty}}(f_{k})\\
	&=\sum\limits_{i=1}^nf_{k}(P_i)\omega^*_{P_i}(1)+\omega^*_{P_{\infty}}(f_{k})\\
	&=\begin{pmatrix}
			{\bf w} & \omega^*_{P_{\infty}}(f_{k})
		\end{pmatrix}\cdot\begin{pmatrix}
					f_{k}(P_1)&\cdots&f_{k}(P_n) & 1
				\end{pmatrix}^T.
\end{align*}
By Lemma \ref{generator}, the proof is completed.
\end{proof}
\begin{corollary}\label{dualex}
Let $\eta$ be a Weil differential satisfying $v_{P_i}(\eta)=-1$ and $\eta_{P_i}(1)=1$ for $1\le i\le n$. Then there exists a basis $\{g_1,\ldots,g_{n-k+1}\}$ of $\mathcal{L}((\eta)+D-(m-1)P_{\infty})$ such that the matrix
\begin{equation*}
	G'_{ex}:=\begin{pmatrix}
		g_1(P_1)& g_1(P_2)&\cdots &g_1(P_n)&0\\
		g_2(P_1)& g_2(P_2)&\cdots &g_2(P_n)&0\\
		\vdots&\vdots&\ddots&\vdots&\vdots\\
		g_{n-k+1}(P_1)& g_{n-k+1}(P_2)&\cdots &g_{n-k+1}(P_n)&-\lambda\\
	\end{pmatrix}
\end{equation*}
constitutes a generator matrix for $C_{ex}(D,mP_{\infty})^{\perp}$, where $\lambda=\sum_{i=1}^{n}f_{k}(P_i)g_{n-k+1}(P_i)$.
\end{corollary}
\begin{proof}
The existence of $\eta$ follows from the Weak Approximation Theorem. By Lemma \ref{isomorphism}, there exists an isomorphism
\[
	\mu:\mathcal{L}((\eta)+D-(m-1)P_{\infty})\to\Omega_F((m-1)P_{\infty}-D)
\]defined via $\mu(f'):=f'\eta$. Let $\{g_1,\ldots,g_{n-k+1}\}$ be a basis of $\mathcal{L}((\eta)+D-(m-1)P_{\infty})$ satisfying
\[
	g_{n-k+1}\in\mathcal{L}((\eta)+D-(m-1)P_{\infty})\setminus\mathcal{L}((\eta)+D-mP_{\infty}).
\]
This implies
\[
	\mu(g_{n-k+1})\in\Omega_F((m-1)P_{\infty}-D)\setminus\Omega_F(mP_{\infty}-D).
\]
For any $f'\in\mathcal{L}((\eta)+D-(m-1)P_{\infty})$, the values $f'(P_i)$ are well-defined since $v_{P_i}(\eta)=-1$. Consequently,
\[
	\mu(f')_{P_i}(1)=f'(P_i)\eta_{P_i}(1)=f'(P_i),
\]
and therefore
\[
	\mu(g_{n-k+1})_{P_{\infty}}(f_{k})=-\sum_{i=1}^{n}f_{k}(P_i)g_{n-k+1}(P_i)=-\lambda.
\]
This demonstrates that $G'_{ex}$ is a parity check matrix for $C_{ex}(D,mP_{\infty})$, therefore the proof is completed.
\end{proof}
We have established a generalization of Proposition \ref{dual} for extended AG codes. The following theorem and corollary address the minimum distance of these extended AG codes.
\begin{theorem}\label{mdex}
If $C_L(D,mP_{\infty})$ has minimum distance $d=n-m$, then $C_{ex}(D,mP_{\infty})$ satisfies $d_{ex}=n-m+1$. Furthermore, let $d'$ and $d'_{ex}$ denote the minimum distances of $C_L(D,mP_{\infty})^{\perp}$ and $C_{ex}(D,mP_{\infty})^{\perp}$, respectively. If $d'=m-(2g-2)$, then $d'_{ex}=d'$.
\end{theorem}
\begin{proof}
For any $f\in\mathcal{L}(mP_{\infty})$ expressed as $f=\sum_{i=1}^{k}a_if_i$, the associated codeword in $C_{ex}(D,mP_{\infty})$ is
\[
	(f(P_1),\ldots,f(P_n),a_{k}).
\]
If $a_{k}=0$, then $f\in\mathcal{L}((m-1)P_{\infty})$. In this case, $f$ has at most $m-1$ zeros in $\{P_1,\ldots,P_n\}$, resulting in a codeword weight of at least $n+1-m$. Conversely, if $f$ has $m$ zeros in $\{P_1,\ldots,P_n\}$, then $a_{k}\neq0$ and the codeword weight equals $n+1-m$. This establishes
\[
	n-m+1\le d_{ex}\le n-m+1.
\]
Given $d'=m-(2g-2)$, there exists a differential $\omega\in\Omega_F(mP_{\infty}-D)$ whose associated codeword
\[
	(\omega_{P_1}(1),\ldots,\omega_{P_n}(1))
\]
has weight $m-(2g-2)$. Consequently, note that $\Omega_F(mP_{\infty}-D)\subset \Omega_F((m-1)P_{\infty}-D)$, we obtain the vector
\[
	(\omega_{P_1}(1),\ldots,\omega_{P_n}(1),0)
\]
also has weight $m-(2g-2)$. Therefore we have
\[
	n+1-\deg((\omega)+D-(m-1)P_{\infty})\le d'_{ex}\le m-(2g-2),
\]
implying $d'_{ex}=m-(2g-2)=d'$. 
\end{proof}
Two immediate corollaries follow from Theorem \ref{mdex} and its proof.
\begin{corollary}
If $C_L(D,mP_{\infty})$ is a $g$-MDS code of length $n$, then $C_{ex}(D,mP_{\infty})$ constitutes a $g$-MDS code of length $n+1$.
\end{corollary}
\begin{proof}
Since $m\ge 2g$ in the settings, we have that the dimension of $C_L(D,mP_{\infty})$ is exactly $m-g+1$, therefore it has minimum distance $d=n-m$ since it is $g$-MDS. Similarly, one can obtain that the minimum distance of $C_L(D,mP_{\infty})^{\perp}$ is $d'=m-(2g-2)$. Then from Theorem~\ref{mdex}, $C_{ex}(D,mP_{\infty})$ is a $g$-MDS.
\end{proof}
\begin{corollary}
The minimum distance $d_{ex}$ of $C_{ex}(D,mP_{\infty})$ satisfies $d_{ex}\ge n-m+1$.
\end{corollary}
\begin{proof}
From the proof of Theorem~\ref{mdex}, we have that
\[
	d_{ex}\ge \min\{d\mid \text{$d$ is the minimum distance of an AG code } C_L(D,mP_{\infty})\}+1.
\]
By the definition of the design distance, we obtain $d_{ex}\ge d^*+1=n-m+1$.
\end{proof}
We now examine the MDS property for extended AG codes.
\begin{proposition}
If both $C_L(D,(m-1)P_{\infty})$ and $C_L(D,mP_{\infty})$ are MDS codes, then $C_{ex}(D,mP_{\infty})$ is MDS.
\end{proposition}
\begin{proof}
For $f\in\mathcal{L}(mP_{\infty})$, the associated codeword in $C_{ex}(D,mP_{\infty})$ is
\[
	(f(P_1),\ldots,f(P_n),a_{k}).
\]
The hypothesis and condition $m\ge 2g$ imply that $C_L(D,(m-1)P_{\infty})$ is an $[n,k-1,n-k+2]$ code. 
\begin{itemize}
	\item For $f\in\mathcal{L}((m-1)P_{\infty})$, there exist at most $k-2$ zeros in $\{f(P_1),\ldots,f(P_n)\}$, yielding a codeword with weight at least $n+1-(k-2)-1=n+1-k+1$. 
	\item For $f\in\mathcal{L}(mP_{\infty})$, the corresponding codeword has weight at least $n+1-(k-1)=n+1-k+1$. 
\end{itemize}
Thus $C_{ex}(D,mP_{\infty})$ achieves the minimum distance $n+1-k+1$, which means that it is an $[n+1,k,n+1-k+1]$ MDS code.
\end{proof}
The following computational example is obtained via MAGMA.
\begin{example}
Let $q=19$ and $\mathcal{X}$ be the elliptic curve $\mathcal{E}:y^2=x^3-x+4$. Note that $\#\mathcal{E}(\mathbb{F}_{19})=23$, which is prime. Fix $P_1:=(0:2:1)$ and define $P_i:=[i]P_1$ for $i=1,\ldots,6$. For any $1\le m\le 5$, the code $C_L(D,mP_{\infty})$ is MDS. Consequently, $C_{ex}(D,mP_{\infty})$ is MDS for $2\le m\le 5$.
\end{example}

We now determine the $\lambda$ when $\mathcal{X}$ is in some special cases, thus the code $C_{ex}(D,mP_{\infty})^{\perp}$ will be given explicitly. Let $t$ be an integer satisfying $n=t\cdot[F:\mathbb{F}_q(x)]$. Suppose that $\{P_1,\ldots,P_n\}$ lie on $\{Q_1,\ldots,Q_t\}$, where each $Q_i$ is a place of $\mathbb{F}_q(x)$ and splits completely in $F/\mathbb{F}_q(x)$, with prime element $x-\alpha_i$. Define $h=\prod_{i=1}^{t}(x-\alpha_i)$, then we have $v_{Q_i}(h)=1$ for $i=1,\ldots,t$. By \cite[Proposition 8.1.2]{Stich}, the differential $dh/h$ of $\mathbb{F}_q(x)$ satisfies $v_{Q_i}(dh/h)=-1$ and $dh/h_{Q_i}(1)=Res_{Q_i}(dh/h)=1$. Denote $\eta=\text{Cotr}_{F/\mathbb{F}_q(x)}(dh/h)$, then we have
\begin{align*}
	(\eta)^F&=\text{Con}_{F/\mathbb{F}_q(x)}(dh/h)+\text{\rm Diff}(F/\mathbb{F}_q(x))\\
	&=(h')^F-\text{\rm Con}_{F/\mathbb{F}_q(x)}(Q_1+\cdots+Q_t+(h)_{\infty})-2[F:\mathbb{F}_q(x)]P_{\infty}\\
	&=(h')^F-D+nP_{\infty}+\text{\rm Diff}(F/\mathbb{F}_q(x))-2[F:\mathbb{F}_q(x)]P_{\infty}\\
	&=(h')^F-D+\text{\rm Diff}(F/\mathbb{F}_q(x))+(n-2[F:\mathbb{F}_q(x)])P_{\infty}
\end{align*}
where $dh=h'dx$. Moreover, we obtain $v_{P_i}(\eta)=-1$ and $\eta_{P_i}(1)=Res_{P_i}(\eta)=1$. Consequently, the $\lambda$ can be determined by calculating $\text{\rm Diff}(F/\mathbb{F}_q(x))$.
\begin{example}
Let $\mathcal{X}$ be a projective line, i.e., $g=0$ with $F=\mathbb{F}_q(x)$. Then $C_L(D,G)$ is exactly an RS code. A basis for $\mathcal{L}(mP_{\infty})$ is $\{1,x,\ldots,x^{m}\}$ with dimension $k=m+1$. We observe that
\[
	\mathcal{L}((\eta)+D-(m-1)P_{\infty})=\mathcal{L}((h')+(n-m-1)P_{\infty}),
\]
which has a basis:
\[
	\left\{\frac{1}{h'},\frac{x}{h'},\ldots,\frac{x^{n-m-1}}{h'}\right\}.
\]
Noting that
\[
	\sum\limits_{i=1}^n\frac{\alpha_i^{n-1}}{h'(\alpha_i)}=1,
\]
we obtain that the generator matrix of $C_L(D,mP_{\infty})^{\perp}$ is given by
\begin{equation*}
	G'_{ex}:=\begin{pmatrix}
		\frac{1}{h'(\alpha_1)}& \frac{1}{h'(\alpha_2)}&\cdots &\frac{1}{h'(\alpha_n)}&0\\
		\frac{\alpha_1}{h'(\alpha_1)}& \frac{\alpha_2}{h'(\alpha_2)}&\cdots &\frac{\alpha_n}{h'(\alpha_n)}&0\\
		\vdots&\vdots&\ddots&\vdots&\vdots\\
		\frac{\alpha_1^{n-m-1}}{h'(\alpha_1)}& \frac{\alpha_2^{n-m-1}}{h'(\alpha_2)}&\cdots &\frac{\alpha_n^{n-m-1}}{h'(\alpha_n)}&-1\\
	\end{pmatrix}.
\end{equation*}
Since $C_L(D,G)$ has minimum distance $d=n-m$, it follows that $C_{ex}(D,G)$ is also MDS. This is just the result in \cite[Theorem 5.3.4]{Huffman-fecc}.
\end{example}
\begin{corollary}
Let $q$ be odd and $\mathcal{X}$ be an elliptic curve, i.e., $g=1$ and $k=m$. Suppose that $n$ is even and the point set $\{P_1,\ldots,P_n\}$, where $P_i=(\alpha_i:\beta_i:1)$ and $P_{i+1}=(\alpha_i:-\beta_i:1)$ with $\beta_i\neq0$ for $i=1,3,\ldots,n-1$. Assume that $\{g_1,\ldots,g_{n-k+1}\}$ is a basis of $\mathcal{L}((n-k+1)P_{\infty})$. Then the matrix
\begin{equation*}
	G'_{ex}:=\begin{pmatrix}
		\frac{g_1}{yh'}(P_1)& \frac{g_1}{yh'}(P_2)&\cdots &\frac{g_1}{yh'}(P_n)&0\\
		\frac{g_2}{yh'}(P_1)& \frac{g_2}{yh'}(P_2)&\cdots &\frac{g_2}{yh'}(P_n)&0\\
		\vdots&\vdots&\ddots&\vdots&\vdots\\
		\frac{g_{n-k+1}}{yh'}(P_1)& \frac{g_{n-k+1}}{yh'}(P_2)&\cdots &\frac{g_{n-k+1}}{yh'}(P_n)&-2\\
	\end{pmatrix}.
\end{equation*}
constitutes a generator matrix for $C_L(D,mP_{\infty})^{\perp}$. Furthermore, if $C_L(D,mP_{\infty})$ is NMDS, then $C_{ex}(D,mP_{\infty})$ is also NMDS.
\end{corollary}
\begin{proof}
Since $\beta_i\neq0$ for $i=1,3,\ldots,n-1$, we have
\[
	(x-\alpha_i)_F=P_i+P_{i+1}-2P_{\infty},
\]
which implies
\[
	(h)_F=D-nP_{\infty}.
\]
Noting that $(dx)_F=(y)_F$ holds for elliptic curves, we consequently derive
\[
	(\eta)_F+D-(k-1)P_{\infty}=(yh')_F+(n-k+1)P_{\infty}.
\]
Furthermore, we can choose $f_k$ and $g_{n-k+1}$ such that
\begin{align*}
f_k=\begin{cases}
x^{\frac{k-3}{2}}y\ &k\ {\rm odd},\\
x^{\frac{k}{2}}\ &k\ {\rm even},
\end{cases}\ \text{and}\ 
g_{n-k+1}=\begin{cases}
x^{\frac{n-k+1}{2}}\ &k\ {\rm odd},\\
x^{\frac{n-k-2}{2}}y\ &k\ {\rm even}.
\end{cases}
\end{align*}
Therefore, we obtain
\begin{align*}
\lambda&=\sum\limits_{i=1}^n\frac{f_kg_{n-k+1}}{yh'}(P_i)\\
&=\sum\limits_{i=1}^n\frac{x^{\frac{n-2}{2}}y}{yh'}(P_i)\\
&=2\sum\limits_{i=1}^t\frac{x^{\frac{n-2}{2}}}{h'}(P_{2i-1})\\
&=2\sum\limits_{i=1}^t\frac{\alpha_i^{t-1}}{h'(\alpha_i)}\\
&=2.
\end{align*}
From Corollary \ref{dualex}, we derive the generator matrix of $C_L(D,mP_{\infty})^{\perp}$. If $C_L(D,mP_{\infty})$ is NMDS, the conditions $d=n-k$ and $d'=k$ directly imply the result by Theorem \ref{mdex}.
\end{proof}
\begin{corollary}
Let $q$ be an odd prime and $\mathcal{H}$ be the Hermitian curve defined over $\mathbb{F}_{q^2}$ with $g=\frac{q(q-1)}{2}$. Suppose that $n=tq$ with $\{P_1,\ldots,P_n\}$ satisfying
\[
	P_{(i-1)q+j}=(\alpha_i:\beta_{i,j}:1)
\]
for $1\le i\le t$ and $0\le j\le q-1$. There exists a basis $\{g_1,\ldots,g_{n-k+1}\}$ of $\mathcal{L}((n-m+2g-1)P_{\infty})$ such that the matrix
\begin{equation*}
	G'_{ex}:=\begin{pmatrix}
		\frac{g_1}{h'}(P_1)& \frac{g_1}{h'}(P_2)&\cdots &\frac{g_1}{h'}(P_n)&0\\
		\frac{g_2}{h'}(P_1)& \frac{g_2}{h'}(P_2)&\cdots &\frac{g_2}{h'}(P_n)&0\\
		\vdots&\vdots&\ddots&\vdots&\vdots\\
		\frac{g_{n-k+1}}{h'}(P_1)& \frac{g_{n-k+1}}{h'}(P_2)&\cdots &\frac{g_{n-k+1}}{h'}(P_n)&-1\\
	\end{pmatrix}
\end{equation*}
constitutes a generator matrix for $C_L(D,mP_{\infty})^{\perp}$.
\end{corollary}
\begin{proof}
It is straightforward to verify that
\[
	(x-\alpha_i)_F=P_{(i-1)q}+P_{(i-1)q+1}+\cdots+P_{iq-1}-qP_{\infty},
\]
and consequently
\[
	(h)_F=D-nP_{\infty}.
\]
Note that $(dx)=(2g-2)P_{\infty}$ holds for Hermitian curves, we consequently derive
\[
	(\eta)_F+D-(m-1)P_{\infty}=(h')_F+(n-m+2g-1)P_{\infty}.
\]
Since $v_{P_{\infty}}(f_{k})=-m$ and $v_{P_{\infty}}(g_{n-k+1})=-(n-m+2g-1)$, it follows that
\[
	v_{P_{\infty}}(f_{k}g_{n-k+1})=-(n+2g-1)=-(tq+q^2-q-1).
\]
We can therefore select a basis for $\mathcal{L}((n-k+2g-1)P_{\infty})$ such that
\[
	f_{k}g_{n-k+1}=x^{t-1}y^{q-1},
\]
which yields
\begin{align*}
\lambda&=\sum\limits_{i=1}^n\frac{f_kg_{n-m+1}}{h'}(P_i)\\
&=\sum\limits_{i=1}^n\frac{x^{t-1}y^{q-1}}{h'}(P_i)\\
&=\sum\limits_{i=1}^t\frac{\alpha_i^{t-1}}{h'(\alpha_i)}\left(\sum\limits_{j=0}^{q-1}\beta_{i,j}^{q-1}\right)\\
&=\sum\limits_{i=1}^t\frac{\alpha_i^{t-1}}{h'(\alpha_i)}\\
&=1.
\end{align*}
By Corollary \ref{dualex}, the proof is completed.
\end{proof}

\subsection{Codes with Length \texorpdfstring{$n+2$}{}}
Let $\delta$ be an arbitrary element in $\mathbb{F}_q$, and define
\[
	\Delta=(0,\ldots,0,1,\delta).
\]
The Roth-Lempel type code $C_{RL,\delta}(D,mP_{\infty})$ is generated by the matrix
\begin{equation*}
	G_{k,RL,\delta}:=\begin{pmatrix}
		G_k & \infty^T & \Delta^T\\
	\end{pmatrix}.
\end{equation*}
\begin{theorem}\label{the.dualrl}
The code $C_{RL,\delta}(D,mP_{\infty})$ is a $[n+2,k]$-linear code over $\mathbb{F}_q$. For any
\[
	\omega^{m-1}\in\Omega_F((m-1)P_{\infty}-D)\setminus\Omega_F(mP_{\infty}-D),
\]
and
\[
	\omega^{m-2}\in\Omega_F((m-2)P_{\infty}-D)\setminus\Omega_F((m-1)P_{\infty}-D),
\]
the matrix
\begin{equation*}
	H_{k,RL,\delta}:=\begin{pmatrix}
		H_k & {\bf 0} & {\bf 0}\\
		{\bf w}^{m-1} & \omega^{m-1}_{P_{\infty}}(f_{k}) &0\\
		{\bf w}^{m-2} & \omega^{m-2}_{P_{\infty}}(f_{k})-\delta\omega^{m-2}_{P_{\infty}}(f_{k-1}) &\omega^{m-2}_{P_{\infty}}(f_{k-1})\\
	\end{pmatrix}
\end{equation*}
serves a parity check matrix for $C_{RL,\delta}(D,mP_{\infty})$ where
\[
	{\bf w}^{m-1}=(\omega^{m-1}_{P_1}(1),\omega^{m-1}_{P_2}(1),\ldots,\omega^{m-1}_{P_n}(1)),
\] 
and
\[
	{\bf w}^{m-2}=(\omega^{m-2}_{P_1}(1),\omega^{m-2}_{P_2}(1),\ldots,\omega^{m-2}_{P_n}(1)).
\] 
\end{theorem}
\begin{proof}
The length and dimension of $C_{RL,\delta}(D,mP_{\infty})$ follow directly from the condition $2g\le m\le n-1$. The differentials $\omega^{m-1}$ and $\omega^{m-2}$ are well-defined, with valuations
\[
	v_{P_{\infty}}(\omega^{m-1})=m-1\ \text{and}\ v_{P_{\infty}}(\omega^{m-2})=m-2.
\]
To complete the proof, it remains to show the following five equations:
\begin{itemize}
	\item $\sum\limits_{i=1}\omega^{m-1}_{P_i}(1)f_{j}(P_i)=0$ and $\sum\limits_{i=1}\omega^{m-2}_{P_i}(1)f_{j}(P_i)=0$ for $1\le j\le k-2$;
	\item $\sum\limits_{i=1}\omega^{m-1}_{P_i}(1)f_{k-1}(P_i)=0$;
	\item $\sum\limits_{i=1}\omega^{m-1}_{P_i}(1)f_{k}(P_i)+\omega^{m-1}_{P_{\infty}}(f_{k})=0$;
	\item $\sum\limits_{i=1}\omega^{m-2}_{P_i}(1)f_{k-1}(P_i)+\omega^{m-2}_{P_{\infty}}(f_{k-1})=0$;
	\item $\sum\limits_{i=1}\omega^{m-2}_{P_i}(1)f_{k}(P_i)+\omega^{m-2}_{P_{\infty}}(f_{k})=0$.
\end{itemize}
The first two equations follow from the valuations $v_{P_{\infty}}(f_j)\le m-1$ for $1\le j\le k-2$. The third equation is derived from the fact that Weil differentials vanish on $F$, which means that:
\begin{align*}
	0&=\omega^{m-1}(f_k)\\
	&=\sum\limits_{i=1}^n\omega^{m-1}_{P_i}(f_k)+\omega^{m-1}_{P_{\infty}}(f_{k})\\
	&=\sum\limits_{i=1}^n\omega^{m-1}_{P_i}(1)f_{k}(P_i)+\omega^{m-1}_{P_{\infty}}(f_{k}).
\end{align*}
Subsequently, the remaining equations can be established in a similar manner. Then the proof is completed.
\end{proof}
\begin{corollary}\label{dualrl}
Let $\eta$ be a Weil differential satisfying $v_{P_i}(\eta)=-1$ and $\eta_{P_i}(1)=1$ for $i=1,\ldots,n$. Then there exists a basis $\{g_1,\ldots,g_{n-k+2}\}$ of $\mathcal{L}((\eta)+D-(m-2)P_{\infty})$ such that the matrix
\begin{equation*}
	G'_{RL,\delta}:=\begin{pmatrix}
		g_1(P_1)& g_1(P_2)&\cdots &g_1(P_n)&0&0\\
		g_2(P_1)& g_2(P_2)&\cdots &g_2(P_n)&0&0\\
		\vdots&\vdots&\ddots&\vdots&\vdots&\vdots\\
		g_{n-k+1}(P_1)& g_{n-k+1}(P_2)&\cdots &g_{n-k+1}(P_n)&-\lambda_1&0\\
		g_{n-k+2}(P_1)& g_{n-k+2}(P_2)&\cdots &g_{n-k+2}(P_n)&-\lambda_3&-\lambda_2\\
	\end{pmatrix}
\end{equation*}
constitutes a generator matrix for $C_{RL,\delta}(D,mP_{\infty})^{\perp}$, where
\begin{align*}
	&\lambda_1=\sum_{i=1}^{n}f_{k}(P_i)g_{n-k+1}(P_i),\\
	&\lambda_2=\sum_{i=1}^{n}f_{k-1}(P_i)g_{n-k+2}(P_i),\\
	&\lambda_3=\sum_{i=1}^{n}f_{k}(P_i)g_{n-k+2}(P_i)-\delta\lambda_2.
\end{align*}
\end{corollary}
\begin{proof}
Following the proof of Corollary \ref{dualex}, there exists an isomorphism
\[
	\mu:\mathcal{L}((\eta)+D-(m-2)P_{\infty})\to\Omega_F((m-2)P_{\infty}-D)
\]
defined via $\mu(f'):=f'\eta$. Let $\{g_1,\ldots,g_{n-k+1},g_{n-k+2}\}$ be a basis of $\mathcal{L}((\eta)+D-(m-2)P_{\infty})$ such that
\[
	g_{n-k+1}\in\mathcal{L}((\eta)+D-(m-1)P_{\infty})\setminus\mathcal{L}((\eta)+D-mP_{\infty}),
\]
and
\[
	g_{n-k+2}\in\mathcal{L}((\eta)+D-(m-2)P_{\infty})\setminus\mathcal{L}((\eta)+D-(m-1)P_{\infty}).
\]
Consequently, we obtain
\[
	\mu(g_{n-k+1})\in\Omega_F((m-1)P_{\infty}-D)\setminus\Omega_F(mP_{\infty}-D),
\]
and
\[
	\mu(g_{n-k+2})\in\Omega_F((m-2)P_{\infty}-D)\setminus\Omega_F((m-1)P_{\infty}-D).
\]
For any $f'\in\mathcal{L}((\eta)+D-(m-1)P_{\infty})$, the values $f'(P_i)$ are well-defined since $v_{P_i}(\eta)=-1$. The values $\lambda_1,\lambda_2,\lambda_3$ are determined by
\[
	\mu(f')_{P_i}(1)=f'(P_i)\eta_{P_i}(1)=f'(P_i).
\]
The conclusion then follows from Theorem \ref{the.dualrl}.
\end{proof}
Let $C_{RL}(D,mP_{\infty})=C_{RL,0}(D,mP_{\infty})$. We establish the following results.
\begin{theorem}\label{mdrl}
Suppose that $C_L(D,mP_{\infty})$ has minimum distance $d=n-m$. If the code generated by the matrix 
\begin{equation*}
	G_1:=\begin{pmatrix}
		f_1(P_1)& f_1(P_2)&\cdots &f_1(P_n)\\
		f_2(P_1)& f_2(P_2)&\cdots &f_2(P_n)\\
		\vdots&\vdots&\ddots&\vdots\\
		f_{k-2}(P_1)& f_{k-2}(P_2)&\cdots &f_{k-2}(P_n)\\
		f_{k}(P_1)& f_{k}(P_2)&\cdots &f_{k}(P_n)\\
	\end{pmatrix}
\end{equation*}
has minimum distance $n-m+1$, then $C_{RL}(D,mP_{\infty})$ satisfies $d_{RL}=n-m+2$; otherwise $d_{RL}=n-m+1$. Furthermore, let $d'=m-(2g-2)$ denote the minimum distance of $C_L(D,mP_{\infty})^{\perp}$. If the code generated by the matrix
\begin{equation*}
	G_2:=\begin{pmatrix}
		g_1(P_1)& g_1(P_2)&\cdots &g_1(P_n)\\
		g_2(P_1)& g_2(P_2)&\cdots &g_2(P_n)\\
		\vdots&\vdots&\ddots&\vdots\\
		g_{n-k}(P_1)& g_{n-k}(P_2)&\cdots &g_{n-k}(P_n)\\
		g_{n-k+2}(P_1)& g_{n-k+2}(P_2)&\cdots &g_{n-k+2}(P_n)\\
	\end{pmatrix}
\end{equation*}
has minimum distance $m-(2g-2)-1$, then $C_{RL}(D,mP_{\infty})$ attains $d'_{RL}=d'$; otherwise $d'_{RL}=d'-1$.
\end{theorem}
\begin{proof}
For any $f\in\mathcal{L}(mP_{\infty})$ expressed as $f=\sum_{i=1}^{k}a_if_i$, the associated codeword in $C_{RL}(D,mP_{\infty})$ is
\[
	(f(P_1),\ldots,f(P_n),a_{k},a_{k-1}).
\]
If $a_k\neq0$ but $a_{k-1}=0$, the minimum distance condition for $G_1$ implies at most $m-1$ zeros in $\{f(P_1),\ldots,f(P_n),a_{k}\}$, resulting in a codeword weight of at least $n-m+2$. If this condition fails, the codeword weight remains at least $n-m+1$. Subsequently, in all remaining cases for $a_k$ and $a_{k-1}$, the weight necessarily exceeds $n-m+2$. 

For the dual code, if $G_2$ satisfies the minimum distance condition, then any codeword
\[
	(f'(P_1),\ldots,f'(P_n),b_{n-k+2},0)
\]
with $f'=\sum_{j=1}^{n-k+2}b_jg_j$ has weight at least $d'$. Applying a similar methodology from Theorem \ref{mdex}, the proof is completed.
\end{proof}
\begin{remark}
Let $F=\mathbb{F}_q(x)$ be the rational function field, and $C_{L,k-1}(D,mP_{\infty})$ denotes the code generated by $G_1$ in Theorem \ref{mdrl}. It follows that $C_{RL}(D,mP_{\infty})$ is MDS if and only if $C_{L,k-1}(D,mP_{\infty})$ is MDS. This observation connects the MDS conditions for $C_{RL}(D,mP_{\infty})$ introduced in \cite{Roth-nRS} and $C_{L,k-1}(D,mP_{\infty})$ investigated in \cite{Li-Zhu-NMDS} (or \cite{HENG2023113538,HanZ24}).

When $F$ is an elliptic function field, $C_{RL}(D,mP_{\infty})$ is clearly AMDS when $C_{L,k-1}(D,mP_{\infty})$ is AMDS. Furthermore, it can be verified that the dual code $C_{RL}(D,mP_{\infty})^{\perp}$ remains AMDS under the same condition. 
\end{remark}
\begin{proposition}
Let $\mathcal{X}$ be an elliptic curve. If $C_L(D,mP_{\infty})$ is MDS, then $C_{RL}(D,mP_{\infty})$ is NMDS.
\end{proposition}
\begin{proof}
It is sufficient to prove that $C_L(D,mP_{\infty})$ means $C_{L,k-1}(D,mP_{\infty})$ and $C_{L,k-1}(D,mP_{\infty})^{\perp}$ are both AMDS. The MDS condition of $C_L(D,mP_{\infty})$ implies that any $f\in\mathcal{L}(mP_{\infty})$ has at most $k-1$ zeros in $\{P_1,\ldots,P_n\}$, while any $g\in\mathcal{L}((\eta)+D-mP_{\infty})$ has at most $n-k-1$ zeros. Since the matrices in Theorem \ref{mdrl} are submatrices of the generator and parity check matrices of $C_L(D,mP_{\infty})$ respectively, it means that the minimum distances of $C_{L,k-1}(D,mP_{\infty})$ and $C_{L,k-1}(D,mP_{\infty})^{\perp}$ are $n-k+1$ and $k+1$ respectively. Note that $C_{L,k-1}(D,mP_{\infty})$ and $C_{L,k-1}(D,mP_{\infty})^{\perp}$ are $[n,k-1]$ and $[n,n-k+2]$ linear codes respectively, thus the proof is completely.
\end{proof}
\section{Covering Radii of Extended Algebraic Geometry Codes}
Using the established notation, we derive that the covering radius of $C_{ex}(D,mP_{\infty})$ has $g+2$ possible values in this section. Additionally, let $C_L(D,mP_{\infty})$ be an elliptic code. If $n$ is sufficiently large or there exists an $[n,k+1]$ MDS elliptic code, we show that there will be $2$ possible values of covering radius of $C_{ex}(D,mP_{\infty})$.

The following lemma provides fundamental bounds for covering radii of extended AG codes:
\begin{lemma}\label{crag}
The covering radius of $C_{ex}(D,mP_{\infty})$ satisfies $n-m-1\le\rho(C_{ex})\le n-m+g$, while the covering radius of its dual code satisfies $m-2g\le\rho(C_{ex}^{\perp})\le m-g+1$.
\end{lemma}
\begin{proof}
Since $C_{ex}(D,mP_{\infty})$ is an $[n+1,k]$ linear code with $k=m-g+1$, the redundancy bound yields:
\begin{align*}
	\rho(C_{ex})&\le n+1-k\\
	&=n+1-(m-g+1)\\
	&=n-m+g.
\end{align*}
Consider a codeword ${\bf v}\in C_{L}(D,(m+1)P_{\infty})\setminus C_{L}(D,mP_{\infty})$. From the designed distance of $C_{L}(D,(m+1)P_{\infty})$, we have $d(C_{L}(D,(m+1)P_{\infty}))\ge n-m-1$, which implies
\[
	d({\bf v}, C_{L}(D,mP_{\infty}))\ge n-m-1.
\]
Let ${\bf\hat v}:=({\bf v},0)$. Then we obtain
\[
	d({\bf\hat v}, C_{ex}(D,mP_{\infty})\ge n-m-1,
\]
establishing $\rho(C_{ex})\ge n-m-1$ by the maximality of the covering radius. For the dual code $C_{ex}(D,mP_{\infty})^{\perp}$ with parameters $[n+1,n-k+1]$, we have
\begin{align*}
	\rho(C_{ex}^{\perp})&\le n+1-(n-k+1)\\
	&=k\\
	&=m-g+1.
\end{align*}
Consider a codeword ${\bf u}\in C_{\Omega}(D,(m-2)P_{\infty})\setminus C_{\Omega}(D,(m-1)P_{\infty})$. From the designed distance of $C_{\Omega}(D,(m-2)P_{\infty})$, we have $d(C_{\Omega}(D,(m-2)P_{\infty}))\ge m-2g$. Similarly, there exists a vector ${\bf\hat u}$ such that
\[
	d({\bf\hat u}, C_{ex}(D,mP_{\infty})^{\perp})\ge m-2g,
\]
proving $\rho(C_{ex}^{\perp})\ge m-2g$.
\end{proof}
\begin{remark}
Let $mathcal{X}$ be a projective line, i.e., $g=0$ with $F=\mathbb{F}_q(x)$. We have $m=k-1$ and consequently $n-k\le\rho(C_{ex})\le n-k+1$. When $n=q$, this result aligns with Conjecture I.2 in \cite{Zhang-Wan-Kaipa-PRS}, which was subsequently proven in \cite{Wuyansheng-when}.
\end{remark}
The following proposition generalizes Theorem 9.1 in \cite{Janwa90}:
\begin{proposition}
Suppose that $C_{ex}(D,mP_{\infty})(\text{or}\ C_{ex}(D,mP_{\infty})^{\perp})$ has minimum distance $d=n-m+1(or\ d'=m-(2g-2))$. If
\[
	n+1=n_q(k,d)(\text{or}\ n+1=n_q(n-k+1,d'),
\]
then $\rho(C_{ex})\in[n-m-1,n-m] (\text{or}\ \rho(C_{ex}^{\perp})\in[m-2g,m-2g+1])$.
\end{proposition}

In the remainder of this subsection, we focus on elliptic curves over $\mathbb{F}_q$ with $q$ odd. Lemma \ref{crag} demonstrates that the covering radius of $C_{ex}(D,mP_{\infty})$ constructed from elliptic curves admits three possible values: $[n-m-1,n-m,m-m+1]$. We shall prove that in certain cases, this range reduces to two possible values.
\begin{theorem}
Let $\mathcal{X}$ be an elliptic curve with $\#\mathcal{X}(\mathbb{F}_q)\ge q+3$. Suppose that $n\ge q+2$ and $k=m\le n-2$. If the MDS conjecture holds, then $\rho(C_{ex})\in[n-k-1,n-k]$ and $\rho(C_{ex}^{\perp})\in[k-2,k-1]$.
\end{theorem}
\begin{proof}
Since $n\ge q+2\ge\lfloor\frac{\#\mathcal{E}(\mathbb{F}_q)}{2}\rfloor+3$, Lemma \ref{maxecmds} implies that $C_L(D,kP_{\infty})$ is NMDS. Consequently, $C_{ex}(D,kP_{\infty})$ has minimum distance $n-k+1$ while its dual $C_{ex}(D,kP_{\infty})^{\perp}$ has minimum distance $k$. Under Conjecture \ref{serges}, no code with parameters $[n,k,n-k+1]$ (or $[n,n-k+1,k]$) exists when $n\ge q+2$. Consequently, we have that $n+1$ is the minimum length admitting codes with these parameters. By Lemma \ref{oplength}, we establish $\rho(C_{ex})\le n-k$ and $\rho(C_{ex}^{\perp})\le k-1$.
\end{proof}

Subsequently, we present some computational examples obtained via MAGMA.
\begin{example}
Let $q=9$ and consider the elliptic curve $\mathcal{E}:=y^2=x^3+x$ over $\mathbb{F}_{9}$ with $\#\mathcal{E}(\mathbb{F}_{9})=4^2$. Suppose that $P_1,\ldots,P_{15},P_{\infty}$ denote all rational places of $F(\mathcal{E})$. Taking $D:=P_1+\ldots+P_{15}$ and $k=9$, then we obtain:
\begin{itemize}
	\item $C_{ex}(D,9P_{\infty})$ is a $[16,9,7]$ code with $\rho(C_{ex})=5$,
	\item $C_{ex}(D,9P_{\infty})^{\perp}$ is a $[16,7,9]$ code with $\rho(C_{ex}^{\perp})=7$.
\end{itemize}
Furthermore, we have $\mathcal{E}(\mathbb{F}_{9})\simeq Z_4\times Z_4$, and we set $y(P_{13})=y(P_{14})=y(P_{15})=0$. Taking $D:=P_1+\ldots+P_{12}$ and $k=9$, then we obtain:
\begin{itemize}
	\item $C_{ex}(D,9P_{\infty})$ is a $[13,9,4]$ code with $\rho(C_{ex})=3$,
	\item $C_{ex}(D,9P_{\infty})^{\perp}$ is a $[13,4,9]$ code with $\rho(C_{ex}^{\perp})=8$.
\end{itemize}
\end{example}
\begin{theorem}
If $C_L(D,(k+1)P_{\infty})$ is an MDS code, then $\rho(C_{ex})\in[n-k,n-k+1]$.
\end{theorem}
\begin{proof}
Consider a codeword
\[
	{\bf v}\in C_{L}(D,(k+1)P_{\infty})\setminus C_{L}(D,kP_{\infty}).
\]
The MDS property of $C_L(D,(k+1)P_{\infty})$ implies
\[
	d({\bf v}, C_{L}(D,kP_{\infty}))\ge n-k.
\]
Let ${\bf\hat v}:=({\bf v},0)$. Then we obtain
\[
	d({\bf\hat v}, C_{ex}(D,kP_{\infty})\ge n-k,
\]
establishing $\rho(C_{ex})\ge n-k$.
\end{proof}
Additional computational examples obtained via MAGMA are presented as follows.
\begin{example}
Following the notations from the previous example, let $\theta$ be a prime element of $\mathbb{F}_q$. Take
\begin{align*}
&P_1:=(1:\theta^2:1),P_2:=(1:\theta^6:1),P_3:=(2:1:1),P_4:=(2:2:1),\\
&P_5:=(\theta:1:1),P_6:=(\theta:2:1),P_7:=(\theta^7:\theta^2:1),P_8:=(\theta^7:\theta^6:1)
\end{align*}
with $D:=P_1+\ldots+P_{8}$. For $C_L(D,kP_{\infty})$ being MDS when $k=1,3,5,7$, we have
\begin{align*}
\rho(C_{ex})=\begin{cases}n-k+1\ \text{if }k=2,\\n-k\ \text{if }k=4,6.
\end{cases}
\end{align*}
and
\begin{align*}
\rho(C_{ex}^{\perp})=\begin{cases}k\ \text{if }k=2,\\k-1\ \text{if }k=4,6.
\end{cases}
\end{align*}
\end{example}

\section{Conclusion and Discussion}
In this paper, we studied dual codes, minimum distances of extended AG codes and Roth-Lempel type AG codes. We presented the explicit structure of their dual codes, and established the bounds on their minimum distances. Furthermore, we demonstrated the covering radii of extended AG codes. In the elliptic case, we improved the general covering radius range to two possible values for two special cases.

Regarding the minimum distance of extended AG codes, we showed that if $C_L(D,mP_{\infty})$ and $C_{\Omega}(D,mP_{\infty})$ exhibit Singleton defect $g$, then $C_{ex}(D,mP_{\infty})$ and $C_{ex}(D,mP_{\infty})^{\perp}$ also have Singleton defect $g$. Consequently, an interesting question is to consider the AG codes from maximal curves:
\begin{problem}
Let $q$ be an odd square. Does there exist $g$-MDS codes from genus $g$ curves with length $q+2g\sqrt{q}$ for arbitrary $g\ge0$? 
\end{problem}
If the answer is true, then we could construct $g$-MDS codes with length $q+1+2g\sqrt{q}$. This raises another question:
\begin{problem}
Is $q+1+2g\sqrt{q}$ the maximal length for all linear $g$-MDS codes?
\end{problem}
This problem shares a similar framework with Conjecture \ref{serges} when $q$ is odd or $4\le k\le q-2+2g\sqrt{q}$.  Based on known databases, the answer also holds for some small finite fields, but a rigorous proof remains to be established.

Furthermore, if both $C_L(D,kP_{\infty})$ and $C_L(D,(k+1)P_{\infty})$ are MDS, then from Lemma II.7 in \cite{Zhang-Wan-Kaipa-PRS}, we directly conclude that $\rho(C_L(D,kP_{\infty}))=n-k$. For the case $g=1$, there exist limited constructions of MDS elliptic codes with consecutive dimensions. Through an extensive analysis of MAGMA examples, we propose the following problem.
\begin{problem}
Consider codes from an elliptic curve. Suppose that both $C_L(D,(k+1)P_{\infty})$ and $C_{ex}(D,kP_{\infty})$ are MDS. Is $\rho(C_{ex})=n-k+1$?
\end{problem}

%\section{Acknowledgment}
%This work is supported by Natural Science Foundation of Guangdong Province (No. 2025A1515011764).

\bibliography{references}
\bibliographystyle{ieeetr}
	
\end{document}